\documentclass{amsart}
\usepackage{amscd,amssymb,subfigure,hyperref,epsfig,amsthm}
\usepackage[arrow,matrix,graph,frame,poly,arc,tips]{xy}
\usepackage{amsmath}

\theoremstyle{plain}
\newtheorem{thm}[subsection]{Theorem}

\newtheorem{lem}[subsection]{Lemma}

\newtheorem{cor}[subsection]{Corollary}

\theoremstyle{definition}
\newtheorem{remark}[subsection]{Remark}
\newtheorem{definition}[subsection]{Definition}
\newtheorem{exm}[subsection]{Example}

\newtheorem{question}[subsection]{Question}
\newtheorem*{ack}{Acknowledgements}

\numberwithin{equation}{section}

\newcommand{\F}{{\mathbb F}}

\newcommand{\Z}{{\mathbb Z}}

\newcommand{\R}{{\mathbb R}}

\newcommand{\Fq}{\F_q}
\newcommand{\Fqstar}{\Fq\sp{\,\times}}
\newcommand{\conv}{\mbox{\rm conv}}
\renewcommand{\P}{{\mathbb P}}

\begin{document}
\title[Remarks on generalized toric codes]{Remarks on generalized toric codes}
\author{John Little}
\date{September 12, 2011}
\email{\href{mailto:little@mathcs.holycross.edu}{little@mathcs.holycross.edu}}
\urladdr{\href{http://mathcs.holycross.edu/~little/homepage.html/}%
{http://mathcs.holycross.edu/\~{}little/homepage.html}}
\address{Department of Mathematics and Computer Science,
College of the Holy Cross,  Worcester, MA 01610}

\begin{abstract}
This note presents some new information on how the minimum distance
of the generalized toric code corresponding to a fixed set
of integer lattice points $S \subset \R^2$ varies with the base field.  The
main results show that in some cases, over sufficiently large
fields, the minimum distance of the code corresponding to a
set $S$ will be the same as that of the code corresponding
to the convex hull $\conv(S)$.  In an example, we will also
discuss a $[49,12,28]$ generalized toric code over $\F_8$,
better than any previously known code according to M. Grassl's
online tables, as of September 2011.
\end{abstract}

\maketitle
\section{Introduction}
We consider linear block codes over finite fields $\Fq$, that
is, vector subspaces $C \subset \Fq^n$.  Following standard
notation in coding theory, $n$ will always denote
the block length, and $k$ will always denote the dimension
of $C$ as a vector space over $\Fq$.  An $[n,k,d]$
code is a linear code with block length $n$, dimension $k$
and minimum distance $d$.  The book \cite{hp} is a good
general reference for other coding theory concepts.

Toric codes are a class of $m$-dimensional cyclic codes introduced
by J. Hansen in \cite{h1}, \cite{h2}.  Hansen's description
uses the geometry of toric varieties corresponding
to polytopes, but toric codes may also be understood
within the general context of evaluation codes.

\begin{definition}
Let $P$ be a rational polytope (the convex hull of a
finite set of integer lattice points), contained in $[0,q-2]^m \subset \R^m$.
Let $\Fq$ be a finite field with primitive element
$\alpha$. For $f \in \Z^m$ with $0\le f_i \le q-2$
for all $i$, let $p_f = (\alpha^{f_1},\ldots,\alpha^{f_m})$
in $(\Fqstar)^m$. For any $e = (e_1,\ldots,e_m) \in P \cap \Z^m$,
let $x^e$ be the corresponding monomial and write
$$(p_f)^e = (\alpha^{f_1})^{e_1}\cdots (\alpha^{f_m})^{e_m} = \alpha^{\langle f,e\rangle}.$$
The toric code over the field
$\Fq$ associated to $P$, denoted by $C_P(\Fq)$, is the linear code of
block length $n = (q-1)^m$ with generator matrix
$$G = ((p_f)^e).$$
The rows are indexed by the $e \in P \cap \Z^m$, and the columns are
indexed by the $p_f \in (\Fqstar)^m$.
In other words, letting $L = \mbox{\rm Span}\{x^e : e \in P\cap \Z^m\}$, we
define the evaluation mapping
\begin{eqnarray*} \mbox{\rm ev} : L & \to & \Fq^{\,(q-1)^m}\\
                                  g & \mapsto & (g(p_f) : p_f \in (\Fqstar)^m)
\end{eqnarray*}
Then $C_P(\Fq) = \mbox{\rm ev}(L)$.
\end{definition}

If $P$ is the interval $[0,\ell-1] \subset \R$, then $C_P(\Fq)$ is the
Reed-Solomon code $RS(\ell,q)$.  So toric codes are, in a sense,
generalizations of Reed-Solomon codes.

In considering code equivalences, the description of dual codes of toric codes,
minimum distance bounds, etc. (see the articles cited below), one is naturally led to consider ``generalized''
toric codes, defined by the same construction, but using an arbitrary set $S$
of integer lattice points in $[0,q-2]^m \subset \R^m$ instead of the set of
all lattice points in a convex polytope.  We will use
the parallel notation $C_S(\Fq)$ for these codes.  If $P = \conv(S)$, then
the code $C_S(\Fq)$ is a subcode of $C_P(\Fq)$.  In algebraic geometry
terms, the $C_S(\Fq)$ can be defined using an incomplete linear system
$V \subset |\mathcal{O}_{X_P}(D_P)|$, where
$X_P$ is the toric variety determined by $P$ and $D_P$ is the
corresponding divisor class on $X_P$.

The survey \cite{mr} covers most of the work on these codes contained
in \cite{j}, \cite{ls}, \cite{r1}, \cite{r2}, \cite{ss}, and \cite{br}.
Not all toric codes or generalized toric codes are as good as Reed-Solomon codes
from the coding theory perspective, but the class of generalized toric codes does
contain some very good codes.  See \S 2 below for a new example.

From now on we will concentrate on the case $m = 2$. 
The principal new results of this note are some observations about
the way the minimum distance of the generalized toric
code $C_S(\Fq)$ depends on $q$ and how it relates to
the minimum distance of $C_P(\Fq)$ if $P = \conv(S)$ and
$q$ is large enough.  These rely
on the connection between the minimum distance of toric
codes and Minkowski sum decompositions of subpolytopes
$Q \subset P$ first noticed by Little and Schenck in \cite{ls}
and later developed and refined by Soprunov and Soprunova
in \cite{ss}.  The other ingredient is some statements about
the distribution of polynomials in $\Fq[u]$ with
given factorization patterns from \cite{c}.

The main result will say that in many
situations, and for large enough $q$, the ``missing''
lattice points in $(P \cap \Z^n) \setminus S$ do not
always help, in the sense that $d(C_{S'}(\Fq))$
can equal  $d(C_S(\Fq))$ for many $S' \supsetneq S$.
In particular, there are even situations where $S' = P \cap \Z^n$
gives a code with the same minimum distance over all $\Fq$
with sufficiently large characteristic.
Read one way, for a fixed polytope $P$, our results
say that the generalized toric codes for $S \subset P$
tend to give interesting results only over small fields.
On the other hand, these results also can help
to identify situations when a proper supercode of
a generalized toric code has the same minimum distance.
Hence they can be used to find improved examples with
the same $d$ but a higher code rate $k/n$.

To conclude, we will also make some remarks about
the toric and generalized toric codes from the ``exceptional triangle'' 
$T_0$ studied in \cite{ss} and the set $S$ obtained by
deleting the interior lattice point.  In this case, when
certain coefficients are nonzero, the polynomials that 
are evaluated to produce the codewords define elliptic curves.
Some facts about numbers of points on elliptic curves
over finite fields provide some interesting results in
this case and provide some extra detail concerning the 
Minkowski decompositions in \cite{ss}.  In this case, 
we will see that the pattern of how $C_S(\Fq)$ depends
on $q$ is considerably more intricate.  

\section{A new best known code}

One of the reasons for the interest in the toric code construction
is that a number of isolated examples of very good codes have
been produced this way.  
 For instance, using certain heuristic
search methods, students at the MSRI-UP 2009 REU program found
\begin{itemize}
\item an $m = 3$ generalized toric code over $\F_5$ with parameters
$[64,8,42]$, and
\item an $m = 2$ generalized toric code over $\F_8$ with
parameters $[49,8,34]$.
\end{itemize}
(Both are reported in \cite{g} and, as of September 2011, are still the best known
codes for these $n,k$ over these fields.)

Here is a new, similar, example.  

\begin{exm}
\label{NewBest4912F8}
Consider the
generalized toric code over $\F_8$ corresponding to the set
$$S = \{(3,0),(4,1),(0,2),(1,2),(2,2),(0,3),(1,3),(4,3),(5,3),(0,4),(2,4),(4,5)\}, $$ 
marked with solid circles in Figure~1.  The polygon
$P$ is the convex hull of $S$, and contains nine lattice points
other than the points of $S$, shown as empty circles.

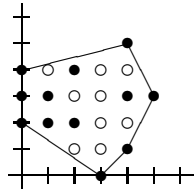
\begin{figure}
\begin{picture}(70,70)(0,0)
\put(5,0){\line(0,1){70}}
\put(0,5){\line(1,0){70}}
\put(15,2){\line(0,1){6}}
\put(25,2){\line(0,1){6}}
\put(35,2){\line(0,1){6}}
\put(45,2){\line(0,1){6}}
\put(55,2){\line(0,1){6}}
\put(65,2){\line(0,1){6}}
\put(15,2){\line(0,1){6}}
\put(2,15){\line(1,0){6}}
\put(2,25){\line(1,0){6}}
\put(2,35){\line(1,0){6}}
\put(2,45){\line(1,0){6}}
\put(2,55){\line(1,0){6}}
\put(2,65){\line(1,0){6}}
\put(35,5){\circle*{4}}
\put(25,15){\circle{4}}
\put(35,15){\circle{4}}
\put(45,15){\circle*{4}}
\put(5,25){\circle*{4}}
\put(15,25){\circle*{4}}
\put(25,25){\circle*{4}}
\put(35,25){\circle{4}}
\put(45,25){\circle{4}}
\put(5,35){\circle*{4}}
\put(15,35){\circle*{4}}
\put(25,35){\circle{4}}
\put(35,35){\circle{4}}
\put(45,35){\circle*{4}}
\put(55,35){\circle*{4}}
\put(5,45){\circle*{4}}
\put(15,45){\circle{4}}
\put(25,45){\circle*{4}}
\put(35,45){\circle{4}}
\put(45,45){\circle{4}}
\put(45,55){\circle*{4}}
\put(5,25){\line(3,-2){30}}
\put(35,5){\line(1,1){10}}
\put(45,15){\line(1,2){10}}
\put(55,35){\line(-1,2){10}}
\put(46,55){\line(-4,-1){40}}
\end{picture}
\caption{The set $S$ and the polygon $P$ in Example~\ref{NewBest4912F8}.}
\end{figure}

There are reducible sections of the line bundle $\mathcal{O}_{X_P}(D_P)$ of the form
$$f(x,y) \equiv y^3x^4(x - \alpha_1)(x - \alpha_2)(x - \alpha_3) \bmod \langle x^7 - 1, y^7 - 1\rangle$$
with the $\alpha_i$ distinct and $\alpha_1 + \alpha_2 + \alpha_3 = 0$.  
Since these have exactly 21 zeroes
at the points in $(\F_8^{\,\times})^2$, we see $d(C_S(\F_8)) \le 28$.  
Using David Joyner's Magma procedures for toric codes (\cite{j}),
it can be checked that this is a $[49,12,28]$ code
over $\F_8$.  In other words, polynomials in two variables
that are linear combinations of the 12 monomials corresponding
to the points in $S$ can have at most $21$ zeroes at the points in
$(\F_8^{\,\times})^2$.  This improves the $[49,12,27]$ code
previously known according to \cite{g}.  $\diamondsuit$
\end{exm}

\section{Factorization patterns for polynomials in one variable}

In this section, we will adapt some known facts about
the distribution of polynomials in $\Fq[u]$ with given
factorization patterns.  The original source for
these statements is \cite{c}; the survey \cite{ghp}
also contains a summary and discussion of the results we need.

Let $q = p^h$ and consider any linear family $\mathcal{F}$ of
polynomials of the form
\begin{equation}
\label{Family}
f(u) = u^\ell + t_1 u^{k_1} + \cdots + t_{m-1} u^{k_{m-1}} + t_m
\end{equation}
in $\Fq[u]$,
where
\begin{enumerate}
\item $p > \ell$,
\item the exponents $\ell > k_1 > \cdots > k_{m-1} > k_m = 0$
are fixed,
\item the coefficients $t_i$, $1 \le i \le m$ run
over the finite field $\Fq$, and
\item the $\ell, k_1, \ldots, k_{m-1}$
are \emph{not} all multiples of some fixed integer $j > 1$.
\end{enumerate}
Some natural questions in this context are:
\begin{itemize}
\item What can be said about the number of elements of the
family $\mathcal{F}$ that are irreducible in $\Fq[u]$?
\item More generally, what can be said about the number
of elements of the family $\mathcal{F}$ that factor in $\Fq[u]$
into a given number of factors of given degrees?
\end{itemize}

To describe the situation for the second question, we will
say that a polynomial $f(u)$ of degree $\ell$ has factorization pattern
$$\lambda = 1^{a_1} 2^{a_2} \cdots \ell^{a_\ell},$$
where $\sum_{i=1}^\ell a_i \cdot i = \ell$, if in $\Fq[u]$,
$f(u)$ factors as a product of $a_i$ irreducible factors of
degree $i$ (not necessarily distinct) for each $i = 1, \ldots, \ell$.
Let
$$T(\lambda) = \frac{1}{a_1! \cdots a_\ell! 1^{a_1} \cdots \ell^{a_\ell}}$$
be the proportion of elements of the symmetric group $S_n$ with
cycle decomposition of shape $\lambda$.  Then S. Cohen proved the following
statement.

\begin{thm}[{\cite[Theorem 3]{c}}]  Let $\mathcal{F}$ satisfy the
conditions above, and let $\mathcal{F}_\lambda$ be the subset
of $\mathcal{F}$ consisting of polynomials with factorization
pattern $\lambda$ in $\Fq[u]$.  Then for all $q$ sufficiently large,
\begin{equation}
\label{Cohen}
|\mathcal{F}_\lambda| = T(\lambda) q^m + O\left(q^{m - \frac{1}{2}}\right)
\end{equation}
where the implied constant depends only on $\ell$.
\end{thm}

For our applications, we want to study factorizations of
shape $\lambda = \lambda_0 := 1^\ell$ where, in addition,
$$f(u) = \prod_{i=1}^\ell (u - \beta_i)$$
with $\beta_i$ distinct in $\Fqstar$.  Now the elements of
the family $\mathcal{F}$ with repeated roots (possibly in
some extension of $\Fq$) correspond to $\Fq$-rational points
$$(t_1,\ldots,t_m) \subset \mathcal{D}_{\mathcal{F}},$$
where
$\mathcal{D}_{\mathcal{F}} = V(\Delta_{\mathcal{F}})$ and
$$\Delta_{\mathcal{F}} = \mathrm{resultant}(f(u),f'(u),u)$$ 
is the discriminant of the family.
Note that $\mathcal{D}_{\mathcal{F}}$ is an $(m-1)$-dimensional affine 
hypersurface, singular and possible
reducible.  However, when the characteristic $p$ is large enough,
it is known by \cite[Theorem 3.1]{fs} that when the two conditions
after (\ref{Family}) hold, 
$\mathcal{D}_{\mathcal{F}}$ can have at most one irreducible 
component other than the hyperplane $V(t_m)$.
By the general bound in \cite[Proposition 12.1]{gl}, it follows that
\begin{equation}
\label{GhorpadeLachaud}
|D_{\mathcal{F}}(\Fq)| \le \delta \pi_{m-1},
\end{equation}
where $\pi_{m-1} = |\P^{m-1}(\Fq)| = q^{m-1} + q^{m-2} + \cdots + q + 1$, and
$\delta = \deg \Delta_{\mathcal F} \le 2\ell - 2$.  (The bound (\ref{GhorpadeLachaud}),
while sufficient for our purposes, is very weak.  Tighter bounds based on
a version of the Weil conjectures for singular varieties can also 
be found in \cite{gl}.)  We have the following
result.

\begin{cor}
\label{BigQ}
If $p > \ell$ and $q = p^h$ is sufficiently
large, there exist elements of the family ${\mathcal F} \subset \Fq[u]$ with factorization
pattern $\lambda_0 = 1^\ell$ in which the irreducible factors
are distinct, and for which all the roots are nonzero.  
\end{cor}

\begin{proof}  The first part of this comes from 
 comparing the orders of growth of the various terms in 
(\ref{Cohen}) and (\ref{GhorpadeLachaud}).
The last part of this is clear since if any of the roots is
zero, then the coefficient $t_m = 0$, and the locus where
that is true has dimension $m - 1$.   
\end{proof}

Note that in particular the conclusion of the Corollary holds for all sufficiently large
primes $p$.   Hence there will be elements of the family ${\mathcal F} \subset 
\F_p[u]$ with factorization pattern $\lambda_0 = 1^\ell$ in which the
irreducible factors are distinct and for which all the roots
are nonzero.  

\begin{remark}
\label{Extension}
If $p$ does not satisfy the condition $p > \ell$, or
if the conclusion of the Corollary does not hold for some $q$, it
is still always possible to find a finite extension of $\Fq$ for
which the statement of the Corollary holds.  Namely, let $f(u)$
be any one element of $\mathcal{F}$ with nonzero discriminant
and nonzero constant term.  If $K$ is a splitting field for
$f(u)$ over $\Fq$, then $f(u)$ splits completely with nonzero
roots in $K[u]$.
\end{remark}  

\begin{exm}  
\label{FamilyEx}
Consider the family $\mathcal{F}$ consisting of
polynomials of the form 
$$u^4 + t_1 u + t_2$$
in $\F_p[u]$ for prime $p$. 
Note that $\mathcal{F}$ contains elements of factorization
pattern $\lambda_0 = 1^4$ with $t_1 = 0$ whenever $p \equiv 1 \bmod 4$, 
since then $\F_p^{\,\times}$ contains 4th roots of unity.  
Doing computations in the Maple computer algebra system,
we found that for all except 5 of the primes $p < 1000$ and
all $p > 19$, there are elements of $\mathcal{F}$ of
factorization pattern $\lambda_0 = 1^4$, and with distinct roots.  
The obvious conjecture
is that there are such polynomials for all $p > 19$.  However,
more precise information about the constants in the asymptotic
result (\ref{Cohen}) than is currently available would be
necessary for a complete proof.  When $p \ge 5$, a constant multiple
of the discriminant of this family can be written as
$$\Delta_{\mathcal{F}} = \left(\frac{t_1}{4}\right)^4 - \left(\frac{t_2}{3}\right)^3.$$ 
The variety 
$\mathcal{D}_{\mathcal{F}} = V(\Delta_{\mathcal{F}})$ is a singular curve of genus 
0 in the $(t_1,t_2)$-plane.  
There are exactly $p - 1$ pairs $(t_1,t_2)$ with 
$t_1t_2 \ne 0$ that make the discriminant equal 0.

One interesting observation
is that the number of polynomials of
factorization pattern $\lambda_0 = 1^4$, and with distinct roots, is
always divisible by $p - 1$.  This follows because the 
mapping $\F_p^{\,\times} \times \mathcal{F} \to \mathcal{F}$ defined by
$$(\beta, f(u)) \mapsto \beta^{-4} f(\beta u)$$
defines an action of $\F_p^{\,\times}$ on $\mathcal{F}$ that preserves the 
factorization pattern, and for which all orbits have order
$p - 1$.  There are similar actions of $\F_q^{\,\times}$ on 
all $\mathcal{F}$ of the form (\ref{Family}) studied here, so this is a 
general phenomenon.  $\diamondsuit$
\end{exm}

\section{Application to generalized toric codes}

In this section we will apply Corollary~\ref{BigQ} to deduce some results
about the minimum distance of generalized toric codes.  First we 
recall the main idea developed in \cite{ls} and \cite{ss}.  Given
a polytope $P$, following \cite{ss}, we define the \emph{full Minkowski
length} of $P$ to be 
$$L(P) = \max \{\ell \mid \exists Q = Q_1 + \cdots + Q_\ell \subseteq P, \dim Q_i > 0 \text{\ all\ } i\},$$
where the addition signs refer to the Minkowski sum of polytopes. 
Theorem 2.6 of \cite{ss} shows that for toric surface
codes ($m = 2$, so $P \subset \R^2$), the full Minkowski length of $P$
is strongly tied to the minimum distance of $C_P(\Fq)$.  
In fact, if $q$ is larger than an explicit lower
bound depending on $L(P)$ and the area of $P$, then the minimum 
distance of the toric code $C_P(\F_q)$ is bounded below as follows:
\begin{equation}
\label{T0possible}
d(C_P(\Fq)) \ge (q - 1)^2 - L(P)(q - 1) - \lfloor 2\sqrt{q}\rfloor + 1,
\end{equation}
and if no maximally decomposable $Q \subset P$ contains an
exceptional triangle term 
(a triangle lattice equivalent to $T_0 = \conv\{(0,0),(1,2),(2,1)\}$),
then 
\begin{equation}
\label{NoT0}
d(C_P(\Fq)) \ge (q - 1)^2 - L(P)(q - 1).
\end{equation}

\begin{exm}
For instance, consider the polygon $P$ from Example~\ref{NewBest4912F8}.  
It can be seen that $L(P) = 6$ and there is a unique $Q \subset P$ with 
6 Minkowski summands, namely the rectangle $Q = \conv\{(0,2),(4,2),(4,4),(0,4)\}$,
which is the Minkowski sum of four primitive lattice segments parallel to the
$x$-axis and two primitive lattice segments parallel to the $y$-axis.  
The corresponding reducible sections of the line bundle $\mathcal{O}_{X_P}(D_P)$
have the form 
$$y^2(x - \alpha_1)(x - \alpha_2)(x - \alpha_3)(x - \alpha_4)(y - \beta_1)(y - \beta_2).$$
So in fact for $q$ large enough, we will have
$$d(C_P(\Fq)) = (q - 1)^2 - 6(q - 1) + 8$$
in this case.  $\diamondsuit$
\end{exm}

Suppose we are in the relatively common case in which the minimum
weight words in the toric code $C_P(\F_p)$ or $C_P(\Fq)$ come by evaluating
polynomials that are linear combinations of monomials
corresponding to a collinear string of $\ell$  
consecutive lattice points in the polytope.
This gives $Q = Q_1 + \cdots + Q_\ell \subset P$ (a Minkowski 
sum of $\ell$ primitive line segments).  
Say the corresponding monomials are 
$$u^a,  \ldots ,  u^{\ell+a}$$  
for some monomial  $u = x^r y^s$ with $\gcd(r,s) = 1$,
and some integers $\ell$
and $a \ge 0$.  The minimum weight codewords then are obtained
by evaluating completely reducible polynomials in $u$:
$$u^a (u - \alpha_1)\cdots (u - \alpha_\ell)$$
where $\alpha_i \in \Fqstar$ are distinct.  
 
Now suppose that we remove some of lattice points
between the endpoints in going to a subset $S \subset P\cap \Z^m$.
The polynomials evaluated to obtain codewords of $C_S(\F_p)$ will
contain linear combinations of some monomials $u^\ell$  and  $u^{k_i}$
with  $\ell > k_1 > \cdots > k_{m-1} > k_m = 0$ 
(after removing the factor $u^a$ that has no zeroes in the torus $(\F_p^{\,\times})^2$).
We obtain polynomials of the form
\begin{equation}
\label{Felt}
f(u) = u^\ell + t_1 u^{k_1} + ... + t_m.
\end{equation}
Note that $\ell \le p - 2$ here by our convention that $P \subset [0,p-2]^2$.  
Hence the condition $p > \ell$ is automatically satisfied.  
In other words, provided that the other conditions on the 
exponents $k_i$ are satisfied, we have elements of a family $\mathcal{F}$ of the 
same form as that considered in (\ref{Family}).  Then Corollary~\ref{BigQ} 
immediately implies the following result.

\begin{thm}
\label{dforSequalsdforP}
Let $P$ be an integral convex polygon in $\R^2$ of full Minkowski
length $L(P) = \ell$.  Suppose in addition that there is a unique $Q \subset P$ 
which decomposes as a sum of $\ell$ nonempty polygons, and that
each of them is a copy of a primitive lattice segment $I$, so $Q = \ell I$.  
Let $S \subset P \cap \Z^2$ satisfy 
\begin{enumerate}
\item $S$ contains the endpoints of $Q$, and
\item The $k_i$ and $\ell$ are not all multiples of any fixed
integer $j > 1$.
\end{enumerate}
Then for all primes $p$ sufficiently large and all $h \ge 1$, letting $q = p^h$,
we have
$$d(C_S(\Fq)) = d(C_P(\Fq)) = (q - 1)^2 - \ell(q - 1).$$
Moreover, for all $q$, there exists $h \ge 1$ such 
that the same statement is true if we replace $q$ by $q^h$.  
\end{thm}
   
Note that hypothesis (2) here rules out the case $m = 1$.  This
is necessary, since polynomials of the form $u^\ell + t_m$ with 
$t_m \ne 0$ have $\ell$ distinct roots only when the field
contains $\ell$th roots of unity, or equivalently when $\ell | (q - 1)$. 
The conclusion of the theorem will also be valid in those cases, 
however.  See Example~\ref{k7Ex} below for some examples. 

\begin{proof}
Corollary~\ref{BigQ} implies that if $p$ is sufficiently 
large the family $\mathcal{F}$ as in (\ref{Felt}) will 
contain elements of factorization pattern $\lambda_0 = 1^\ell$
with distinct nonzero roots.  The corresponding polynomials
in $x,y$ obtained by substituting $u = x^r y^s$ will have
the same sort of factorization.  Since $\gcd(r,s) = 1$, 
each factor $x^r y^s - \alpha_i$
has exactly $p - 1$ zeroes in $(\F_q^{\,\times})^2$ and the
sets of roots for distinct $\alpha_i$ are disjoint.  Therefore
for all sufficiently large $p$, the toric code will contain
words of weight $(p - 1)^2 - \ell(p - 1)$, and the minimum
distance satisfies
$$d(C_S(\F_p)) \le (p - 1)^2 - \ell(p - 1).$$
The generalized toric code $C_S(\F_p)$ is a subcode of
the code from the full polygon $P$.  The lower bound from \cite{ss}
quoted above in (\ref{NoT0}) gives the reverse inequality
and the equality claimed in the statement follows for all $p$
sufficiently large.  The final part here follows by 
Remark~\ref{Extension}.
\end{proof}

Related statements along the lines 
of Theorem~\ref{dforSequalsdforP} apply in different
situations depending on the shape of the $Q$ giving the
maximally Minkowski-decomposable subpolygon of $P$.  
However, we will not pursue them here.  

\begin{exm}
\label{k7Ex}
Consider the generalized toric codes $C_S(\F_q)$ for 
the set $S$ indicated with solid circles in Figure~2.

The polygon $P$ here is the quadrilateral
$P = \conv\{(0,0),(2,0),(3,1),(1,4)\}$.  It can be seen that
the full Minkowski length of $P$ is $L(P) = 4$, and $P$ contains
just one Minkowski sum of 4 indecomposable polygons, namely
the line segment $Q = \conv\{(1,0),(1,4)\}$.  By the results
of Example~\ref{FamilyEx}, and Theorem~\ref{dforSequalsdforP},
we expect that 
$$d(C_S(\F_p)) = d(C_P(\F_p)) = (p - 1)^2 - 4(p - 1)$$
for all $p > 19$ and all $q = p^s$ for $s \ge 4$.  The following
results of Magma computations indicate what happens for $p \le 19$,
and also for the prime powers $7 \le q \le 19$:
\begin{eqnarray*}
d(C_S(\F_7)) = 18 & {\rm vs.} & 6^2 - 4\cdot 6 = 12\\
d(C_S(\F_8)) = 33 & {\rm vs.} & 7^2 - 4\cdot 7 = 21\\
d(C_S(\F_9)) = 32 & {\rm vs.} & 8^2 - 4\cdot 8 = 32\\
d(C_S(\F_{11})) = 70 & {\rm vs.} & 10^2 - 4\cdot 10  = 60\\
d(C_S(\F_{13})) = 96 & {\rm vs.} & 12^2 - 4\cdot 12  = 96\\
d(C_S(\F_{16})) = 165 & {\rm vs.} & 15^2 - 4\cdot 15  = 165\\
d(C_S(\F_{17})) = 192 & {\rm vs.} & 16^2 - 4\cdot 16  = 192\\
d(C_S(\F_{19})) =  270 & {\rm vs.} & 18^2 - 4\cdot 18  = 252.
\end{eqnarray*}

Over $\F_7$, there are no polynomials $u^4 + t_1 u + t_2$
with $t_2 \ne 0$ and factorization pattern $\lambda_0 = 1^4$.  
Since $7 \equiv 1 \bmod 3$, however, $\F_7$ contains cube roots of unity.
Hence, minimum weight codewords in that case come from
polynomials of the form $xy(y^3 - \beta)$.  

Over $\F_8$, the
minimum weight codewords come from genus 5 curves with 16
$\F_8$-rational points in the torus $(\F_8^{\,\times})^2$.  

The codes over $\F_9,\F_{13}$, and $\F_{17}$ all have 
$d = (q - 1)^2 - 4(q - 1)$.   
Since these $q$ all satisfy $q \equiv 1 \bmod 4$, 
these fields contain 4th roots of unity and hence there
are $u^4 + t_1 u + t_2$ with 
factorization pattern $\lambda_0 = 1^4$
with $t_1 = 0$ and $t_2 \ne 0$.  Minimum weight codewords
come from polynomials of the form
$x(y^4 - \beta)$ with $\beta \ne 0$.    

The code over $\F_{11}$ illustrates
the fact that while there are polynomials
$u^4 + t_1 u + t_2$ with factorization pattern $\lambda_0 = 1^4$
over this field, all such polynomials have a repeated root.
Therefore the minimum codeword weight is 
$(q - 1)^2 - 3(q - 1)$ rather than $(q - 1)^2 - 4(q-1)$.

The code over $\F_{16}$ illustrates
the comment from Remark~\ref{Extension}.  Note
that $\F_{16}$ is the splitting field 
of the polynomial $y^4 + y + 1$ over $\F_2$. Hence we
obtain words of weight $15^2 - 4\cdot 15$ in this code too.

\begin{figure}
\begin{picture}(70,70)(0,0)
\put(5,0){\line(0,1){70}}
\put(0,5){\line(1,0){70}}
\put(15,2){\line(0,1){6}}
\put(25,2){\line(0,1){6}}
\put(35,2){\line(0,1){6}}
\put(45,2){\line(0,1){6}}
\put(55,2){\line(0,1){6}}
\put(65,2){\line(0,1){6}}
\put(15,2){\line(0,1){6}}
\put(2,15){\line(1,0){6}}
\put(2,25){\line(1,0){6}}
\put(2,35){\line(1,0){6}}
\put(2,45){\line(1,0){6}}
\put(2,55){\line(1,0){6}}
\put(2,65){\line(1,0){6}}
\put(5,5){\circle*{4}}
\put(15,5){\circle*{4}}
\put(25,5){\circle*{4}}
\put(15,15){\circle*{4}}
\put(25,15){\circle{4}}
\put(35,15){\circle*{4}}
\put(15,25){\circle{4}}
\put(25,25){\circle*{4}}
\put(15,35){\circle{4}}
\put(15,45){\circle*{4}}
\put(5,5){\line(1,0){20}}
\put(25,5){\line(1,1){10}}
\put(35,15){\line(-2,3){20}}
\put(15,45){\line(-1,-4){10}}
\end{picture}
\caption{The set $S$ and the polygon $P$ in Example~\ref{k7Ex}.}
\end{figure}
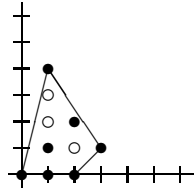
Finally, the code over $\F_{19}$ has some words of
weight $270$ from evaluation of polynomials $xy(y^3 - \beta)$
since $19 \equiv 1 \bmod 3$.  (There are also several
polynomials $y^4 + t_1 y + t_2$ that have factorization
patterns $\lambda_0 = 1^4$.  However all of the polynomials
that do factor that way have repeated roots.) 
$\diamondsuit$
\end{exm}

\section{The exceptional triangle}

We will now consider what happens for polygons $P$ where
a maximal Minkowski-decomposable subpolygon $Q$ has a
Minkowski decomposition involving the 
triangle $T_0 = \conv\{(0,0),(1,2),(2,1)\}$.  This is 
affine equivalent to the triangle used in \cite{ss} (see
Figure~3).  We will use this form because of its relation
to the well-known Hessian family of elliptic curves.  

This is the only case in the plane where a Minkowski-indecomposable
polygon contains an interior lattice point, namely $(1,1)$.  Hence
we will begin by comparing $d(C_{T_0}(\Fq))$ and $d(C_S(\Fq))$
where $S = \{(0,0),(1,2),(2,1)\}$, omitting the interior
lattice point.  The  presence of the interior lattice point in 
$T_0 \cap \Z^2$ shows that we are considering curves of (arithmetic) 
genus 1.  Hence the theory of elliptic curves will be important.
The facts about elliptic curves over finite fields
that we will need can be found in \cite{s} and \cite{h}.

The projective completions of the curves defined by 
the linear combinations of monomials corresponding
to the lattice points in $T_0$ form the family of cubic 
curves defined by homogeneous equations:
\begin{equation}
\label{cubics}
a x^2y + b x y^2  + c x y z + d z^3  = 0.
\end{equation}
If at least one of the coefficients $a,b,d$ vanishes, there
are at most $q-1$ affine $\Fq$-rational points with nonzero
coordinates on the corresponding curve.   It follows that 
$$d(C_{T_0}(\Fq)) \le (q - 1)^2 - (q - 1)$$
and similarly
\begin{equation}
\label{Sbound}
d(C_S(\Fq)) \le (q - 1)^2 - (q - 1).
\end{equation}
Hence, we want to concentrate on the cases with $abd \ne 0$,
but $c$ possibly $0$ in the case of $T_0$, 
and $c = 0$ for the case of $S$.  Thinking along
the lines of our earlier results, we pose the following
question.

\begin{figure}
\begin{picture}(70,70)(0,0)
\put(5,0){\line(0,1){70}}
\put(0,5){\line(1,0){70}}
\put(15,2){\line(0,1){6}}
\put(25,2){\line(0,1){6}}
\put(35,2){\line(0,1){6}}
\put(45,2){\line(0,1){6}}
\put(55,2){\line(0,1){6}}
\put(65,2){\line(0,1){6}}
\put(15,2){\line(0,1){6}}
\put(2,15){\line(1,0){6}}
\put(2,25){\line(1,0){6}}
\put(2,35){\line(1,0){6}}
\put(2,45){\line(1,0){6}}
\put(2,55){\line(1,0){6}}
\put(2,65){\line(1,0){6}}
\put(5,5){\circle*{4}}
\put(15,25){\circle*{4}}
\put(25,15){\circle*{4}}
\put(15,15){\circle{4}}
\put(5,5){\line(1,2){10}}
\put(5,5){\line(2,1){20}}
\put(25,15){\line(-1,1){10}}
\end{picture}
\caption{The exceptional triangle $T_0$ and the set $S$.}
\end{figure}
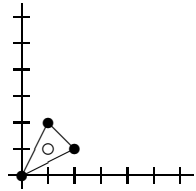
\begin{question}
\label{T0Question}
For sufficiently large primes $p$, or sufficiently high powers
$q = p^h$, is 
$$d(C_S(\Fq)) = d(C_{T_0})?$$
Equivalently, do we expect curves in the 
family (\ref{cubics}) with  $c = 0$  to achieve the maximum number of  
$\F_p$- or $\Fq$-rational points (among curves in the family)?
\end{question}

\begin{exm}
\label{T0vsS}
Some experimentation using Magma reveals that the answer to 
this question is not at all clear at first.  As in the discussion
above, $T_0$ is the exceptional triangle, and $S$ is
the set of vertices, omitting the interior lattice point.
\begin{eqnarray*}
d(C_S(\F_5)) = 12 & {\rm vs.} & d(C_{T_0}(\F_5)) = 10\\
d(C_S(\F_7)) = 27 & {\rm vs.} & d(C_{T_0}(\F_7)) = 27\\
d(C_S(\F_8)) = 42 & {\rm vs.} & d(C_{T_0}(\F_8)) = 40\\
d(C_S(\F_9)) = 56 & {\rm vs.} & d(C_{T_0}(\F_9)) = 52\\
d(C_S(\F_{11})) = 90 & {\rm vs.} & d(C_{T_0}(\F_{11})) = 85\\
d(C_S(\F_{13})) = 126 & {\rm vs.} & d(C_{T_0}(\F_{13})) = 126\\
d(C_S(\F_{16})) = 207 & {\rm vs.} & d(C_{T_0}(\F_{16})) = 204\\
d(C_S(\F_{17})) = 240 & {\rm vs.} & d(C_{T_0}(\F_{17})) = 235\\
d(C_S(\F_{19})) = 300 & {\rm vs.} & d(C_{T_0}(\F_{19})) = 300\\
d(C_S(\F_{23})) = 462 & {\rm vs.} & d(C_{T_0}(\F_{23})) = 454.
\end{eqnarray*}
\noindent
Note that the minimum distance of the code from $S$ is sometimes
greater, and sometimes the same as the
minimum distance of the code from $T_0$.  $\diamondsuit$
\end{exm}

We will see that there are arbitrarily large $q$ for which
$d(C_S(\Fq)) > d(C_{T_0}(\Fq))$.  To prove this,
we begin with some general observations.  Much of this has
been noted before in Theorem 2 of \cite{br}, but not in the
context of the possible Minkowski decompositions studied in \cite{ss}.   
Since we assume $abd \ne 0$, to normalize, we will take $d = 1$.  

There are three distinct points on the curve on the line 
at infinity $z = 0$ in all cases  (that is, whether or not $c = 0$).  
All three of them are flexes, and the lines  
$x = 0$, $y = 0$, $a x + b y + c z = 0$ are the inflectional 
tangents.  So the points in the torus $(\Fqstar)^2$ 
are all the affine points of the curve.  

If $abd \ne 0$ but  $c = 0$,  then the curve is smooth of genus 1.  
This can be seen since the system defining the Jacobian ideal:
\begin{align*}
                  2 a x y + b y^2 &=  y(2 a x + b y) = 0\\                   
                  a x^2 + 2 b x y &=  x(a x+ 2 b y) = 0 \\                  
                  3 z^2 &=  0.
\end{align*}
implies $z = 0$.  But, by (1), the curve must be smooth, since 
all three points on the line at infinity are smooth points.   

If $c \ne 0$, there are also singular (nodal) 
cubics in the family.  In this case the Jacobian system has the solution
$$x = \frac{-c}{3a},\quad  y = \frac{-c}{3b},\quad  z = 1.$$   
Substituting into the equation of the curve, we get:
$$\frac{-c^3}{27ab} - \frac{c^3}{27ab} + \frac{c^3}{9ab} + 1 = 0.$$ 
so if  $c^3 = - 27 a b$, the curve has a singular point at the point
with homogeneous coordinates
$$(x:y:z) = \left(\frac{-c}{3a}: \frac{-c}{3b}: 1\right).$$

If the origin of the group structure on the points of a smooth cubic
curve is placed at an inflection point, the other 
inflection points are points have order 3.   Since
the 3-torsion points form a subgroup of the group of $\Fq$-rational
points, we have the following statement.

\begin{lem}
\label{divby3}
For all the smooth elements $E$ of the family (\ref{cubics}) over
$\Fq$, the number of $\Fq$-rational points is \emph{divisible by 3}.
\end{lem}

The discussion of \S 4.2 of \cite{h} shows, in fact, that the
family (\ref{cubics}) is a sort of universal family for
elliptic curves over $\Fq$ with nontrivial 3-torsion subgroups.  Every
isomorphism class of such curves is represented by some element of
our family. 

If $p\ge 3$, by an easy change of coordinates, 
the equations (\ref{cubics}) can be put into Weierstrass form.
Namely, dehomogenize with respect to $x$, and complete the
square in $y$.  If $p \ge 5$, by a further change of coordinates, 
the Weierstrass form can be taken to 
$$u^2 = v^3 + A v + B.$$
If $c = 0$ to start, then after this change of coordinates
it will be true that $A = 0$.  

The $j$-invariant of an elliptic curve in this form is 
$$j = 1728 \frac{4 A^3}{4 A^3 + 27 B^2}.$$ 
Hence $j = 0$ if and only if $A = 0$.  .

When  $q \equiv 2 \bmod 3$ for an odd prime power $q$,  
elliptic curves with $j = 0$ are \emph{supersingular} 
elliptic curves (see Proposition 4.31 of \cite{was}). 
There are many equivalent 
characterizations of this property and it follows that 
$$|E(\F_{q^h})| = \begin{cases}
                   q^h + 1 & h \text{\ odd}\\
                   q^h + 1 + 2q^{h/2} & \text{\ if\ } h \equiv 2 \bmod 4\\
                   q^h + 1 - 2q^{h/2} & \text{\ if\ } h \equiv 0 \bmod 4.
                  \end{cases}
$$
In other words, supersingular elliptic curves defined over
$\Fq$ achieve the Hasse-Weil \emph{upper} bound over $\F_{q^h}$ when
$h \equiv 2 \bmod 4$.  On the other hand, they achieve the Hasse-Weil
\emph{lower} bound over $\F_{q^h}$ when $h \equiv 0\bmod 4$. 

The above observations show that the answer to our Question~\ref{T0Question} 
is negative, because of some subtle arithmetic
facts concerning the numbers of points on certain elliptic
curves!  Some of the following reasoning also appears in the proof of
Theorem 2 in \cite{br}.  

\begin{thm}
\label{T0Sth}
Let $q$ be odd and $q \equiv 2\bmod 3$.  Then
$$d(C_S(\Fq)) = (q - 1)^2 - (q - 1) > d(C_{T_0}(\Fq)).$$
\end{thm}

\begin{proof}  If
$q$ is odd and $q \equiv 2 \bmod 3$, then because the 
corresponding elliptic curves are supersingular (and recalling that we
must subtract the three points at infinity) all of the
codewords of $C_S(\F_q)$ obtained from evaluation of 
$axy^2 + bxy^2 + d$ with $abd \ne 0$
will have weight 
$$(q - 1)^2 - (q + 1 - 3) > (q - 1)^2 - (q - 1).$$
On the other hand, by (\ref{Sbound}) there are also codewords
of weight $(q - 1)^2 - (q - 1)$ from polynomials with one 
coefficient equal to zero.   Those give 
the minimum weight words in this case.  

On the other hand, we need to determine the minimum distance
of $C_{T_0}(\Fq)$.  By the theorem of Waterhouse (Theorem 4.1 of \cite{wat}), 
we know that there are elliptic curves over $\Fq$ with 
$$|E(\Fq)| = q + 1 + t$$
for all integers $t$ with $t \le \lfloor 2 \sqrt{q}\rfloor$ and 
$\gcd(t,q) = 1$ (as well as some other
possibilities).  By Lemma~\ref{divby3} and the universality of
our family (\ref{Family}) for curves with nontrivial 3-torsion,
there will be curves here with $q + 1 + t$ points rational
over $\Fq$ if $t$ is the \emph{largest} integer satisfying
$t \le \lfloor 2\sqrt{q}\rfloor$, $t$ 
prime to $q$, and such that $3|(q + 1 + t)$.  
These give codewords of considerably
smaller weight, in some cases close to 
$$(q - 1)^2 - (q + 1 + 2\sqrt{q} - 3).$$  
So the minimum distance of the code from $S$
will be strictly larger than that of the code from $T_0$ for
all such $q$.  
\end{proof}

\begin{exm}
\label{tieup}
For instance, refer again to Example~\ref{T0vsS}.  
With $p = 17 \equiv 2 \bmod 3$, Theorem~\ref{T0Sth} applies.  
The largest $t$ such that $t \le \lfloor 2\sqrt{17}\rfloor $ and $3|(17 + 1 + t)$ is 
$t = 6$.  There are elliptic curves in the family (\ref{Family}) with 
24 $\F_{17}$-rational points and hence 21 points in $(\F_{17}^{\,\times})^2$
disregarding the three points at infinity.  This gives
$d(C_{T_0}(\F_{17})) = 16^2 - 21 = 235$.
However, $d(C_S(\F_{17})) = 16^2 - 16 = 240$.  

On the other hand, when $q$ is odd and $\equiv 1 \bmod 3$, 
the elliptic curves with $j = 0$ are not supersingular.  And in
fact the minimum weight words in $C_{T_0}(\Fq)$ sometimes 
come from $C_S(\Fq)$ in this case.  $\diamondsuit$
\end{exm}

We conclude with a final remark.  
Phenomena similar to those seen in Theorem~\ref{T0Sth}, in 
which $d(C_S(\Fq)) > d(C_P(\Fq))$ can occur for arbitrarily
large $q$, also
occur in polygons for which the maximal
Minkowski-reducible subpolygon $Q$ contains a $T_0$ summand.
\begin{figure}
\begin{picture}(70,70)(0,0)
\put(5,0){\line(0,1){70}}
\put(0,5){\line(1,0){70}}
\put(15,2){\line(0,1){6}}
\put(25,2){\line(0,1){6}}
\put(35,2){\line(0,1){6}}
\put(45,2){\line(0,1){6}}
\put(55,2){\line(0,1){6}}
\put(65,2){\line(0,1){6}}
\put(15,2){\line(0,1){6}}
\put(2,15){\line(1,0){6}}
\put(2,25){\line(1,0){6}}
\put(2,35){\line(1,0){6}}
\put(2,45){\line(1,0){6}}
\put(2,55){\line(1,0){6}}
\put(2,65){\line(1,0){6}}
\put(5,5){\circle*{4}}
\put(15,25){\circle*{4}}
\put(25,15){\circle*{4}}
\put(15,15){\circle{4}}
\put(25,25){\circle*{4}}
\put(35,15){\circle*{4}}
\put(15,5){\circle*{4}}
\put(5,5){\line(1,2){10}}
\put(5,5){\line(2,1){20}}
\put(25,15){\line(-1,1){10}}
\put(15,5){\line(2,1){20}}
\put(35,15){\line(-1,1){10}}
\put(25,25){\line(-1,0){10}}
\end{picture}
\caption{A Minkowski sum with the exceptional triangle.}
\end{figure}
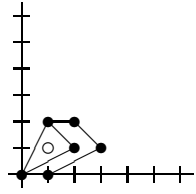

\begin{exm}
\label{MinkSum}
Consider the Minkowski sum $P = T_0 + I$, where $I$ is the 
interval $I =  \conv\{(0,0),(1,0)\}$, shown
in Figure~4.  We study the generalized toric codes for
$$S = (P \cap \Z^2) \setminus \{(1,1)\}$$
obtained by removing one of the two interior lattice points
from $P$.  

For all odd $q \equiv 2 \bmod 3$, we will again have 
$d(C_S(\Fq)) > d(C_P(\Fq))$. 
The minimum distances of the two codes over
small fields are as follows.

\begin{eqnarray*}
d(C_S(\F_7)) = 22 & {\rm vs.} & d(C_P(\F_7)) = 21\\
d(C_S(\F_8)) = 36 & {\rm vs.} & d(C_P(\F_8)) = 33\\
d(C_S(\F_9)) = 48 & {\rm vs.} & d(C_P(\F_9)) = 44\\
d(C_S(\F_{11})) = 80 & {\rm vs.} & d(C_P(\F_{11})) = 75\\
d(C_S(\F_{13})) = 114 & {\rm vs.} & d(C_P(\F_{13})) = 114\\
d(C_S(\F_{16})) = 192 & {\rm vs.} & d(C_P(\F_{16})) = 189\\
d(C_S(\F_{17})) = 224 & {\rm vs.} & d(C_P(\F_{17})) = 219\\
d(C_S(\F_{19})) = 282 & {\rm vs.} & d(C_P(\F_{19})) = 282.
\end{eqnarray*}

Note that $P$ contains the 
two term Minkowski sum $\conv(\{(1,0),(1,2)\})$,
as well as several other Minkowski decomposable parallelograms.
If instead of this $S$ we consider $R$ consisting of the $5$ noninterior
lattice points in $P$, then codewords obtained by evaluating reducible linear
combinations of the corresponding monomials have the 
minimum possible weights 
$$(q - 1)^2  - 2(q - 1)$$
in many of these cases, since there are reducible
polynomials of the form 
$$x(y - \alpha_1)(y - \alpha_2)$$
with $\alpha_1 \ne \alpha_2$ and $\alpha_1 + \alpha_2 = 0$  whenever $q$ is odd. 
These have $2(q - 1)$ zeroes in $(\Fq^{\,\times})^2$. $\diamondsuit$ 
\end{exm}

\begin{ack}
The work reported here began during a visit to the University
of Illinois at Urbana-Champaign, hosted by Hal Schenck.  The author
would like to thank him for his hospitality and valuable comments.
\end{ack} 

\bibliographystyle{amsalpha}

\end{document}